\documentclass[11pt]{article}
\title{The number of distinct distances from a vertex of a convex polygon}
\date{March 22, 2013}
\author{Gabriel Nivasch\footnote{Ariel University, Ariel, Israel. Work was done when the author was at EPFL, Lausanne, Switzerland.}\\\texttt{gabrieln@ariel.ac.il}
\and J\'anos Pach\footnote{EPFL, Lausanne, Switzerland and R\'enyi Institute, Budapest, Hungary. Research partially supported by Swiss National Science Foundation Grants 200021-137574 and 200020-144531, by Hungarian Science Foundation Grant OTKA NN 102029 under the EuroGIGA programs ComPoSe and GraDR, and by NSF grant CCF-08-30272.}\\\texttt{pach@cims.nyu.edu}
\and Rom Pinchasi\footnote{Mathematics Department, Technion---Israel Institute of Technology, Haifa 32000, Israel.}\\\texttt{room@math.technion.ac.il} \and Shira Zerbib\footnotemark[\value{footnote}]\\\texttt{zgshira@tx.technion.ac.il}
}
\usepackage{amsmath,amssymb,amsthm,color,graphicx}

\newtheorem{theorem}{Theorem}
\newtheorem{lemma}[theorem]{Lemma}

\newtheorem{corollary}[theorem]{Corollary}

\theoremstyle{definition}
\newtheorem{definition}[theorem]{Definition}
\newtheorem{problem}{Problem}

\newcommand{\comment}[1]{}
\newcommand{\ignore}[1]{}

\def\clap#1{\hbox to 0pt{\hss#1\hss}}

\makeatletter
  \def\moverlay{\mathpalette\mov@rlay}
  \def\mov@rlay#1#2{\leavevmode\vtop{%
    \baselineskip\z@skip \lineskiplimit-\maxdimen
    \ialign{\hfil$#1##$\hfil\cr#2\crcr}}}
\makeatother

\def\mangle{{\measuredangle}}

\newcommand{\remove}[1]{}

\begin{document}

\maketitle

\begin{abstract}
Erd\H os conjectured in 1946 that every $n$-point set $P$ in convex position in the plane contains a point that determines at least $\lfloor n/2 \rfloor$ distinct distances to the other points of $P$. The best known lower bound due to
Dumitrescu (2006) is $13n/36 - O(1)$. In the present note, we slightly improve
on this result to $(13/36 + \varepsilon)n - O(1)$ for $\varepsilon \approx 1/23000$. Our main ingredient is an improved bound on the maximum number of isosceles triangles determined by $P$.
\end{abstract}

\section{Introduction}

We say that a point set $P$ {\em determines} a distance $d$ if $P$ contains two elements such that their Euclidean distance is $d$. Given a positive integer $n$, what is the maximum number $g(n)$ such that every set of $n$ points in the plane determines at least $g(n)$ distinct distances? According to a famous conjecture of Erd\H os \cite{E46}, we have
$g(n)=\Omega(\frac{n}{\sqrt{\log n}})$. The number of distinct distances determined by a $\sqrt{n}\times\sqrt{n}$ piece of the integer lattice is $O(\frac{n}{\sqrt{\log n}})$ \cite{E46}, which shows that his conjecture, if true, would be essentially best possible.

In a recent breakthrough, Guth and Katz \cite{GK11} have come very close to proving Erd\H os's conjecture. They showed that $g(n) = \Omega(\frac{n}{\log n})$. This is a substantial improvement on the previous bound of $g(n)\ge n^{0.864\ldots}$ by Katz and Tardos \cite{KaT04}, which was the last step in a long series of successive results \cite{Mo52}, \cite{Ch84}, \cite{ChST92}, \cite{Sz93}, \cite{SoT01}, \cite{Ta03}.

In the same paper, Erd\H os \cite{E46} also made a much stronger conjecture. Let $f(n)$ denote the maximum number such that every set of $n$ points in the plane contains a point from which there are at least $f(n)$ distinct distances to the other $n-1$ points of the set.  Clearly, we have $f(n) \leq g(n)$. Erd\H os conjectured that
in fact $f(n) = \Omega (\frac{n}{\sqrt{\log n}})$. This conjecture is still wide open, although all the above mentioned lower bounds, with the exception of the one due to Guth and Katz, also apply to $f(n)$. In particular, the best known lower bound of $f(n)$ is still $f(n)\ge n^{0.864\ldots}$ by Katz and Tardos.

As Erd\H os suggested, the same question can be studied for point sets with special properties. We say that $n$ points in the plane are in {\em convex position} if they form the vertex set of a convex $n$-gon. A set of $n$ points is in {\em general position} if no $3$ of its elements are collinear. Let $f_{\rm conv}(n)$ (and $f_{\rm gen}(n)$) denote the largest number such that every set of $n$ points in the plane in convex (resp., in general) position in the plane contains a point from which there are at least these many distinct distances to the remaining $n-1$ points. Since every set in convex position is also in general position, we have  $f_{\rm gen}(n) \leq f_{\rm conv}(n)$. The vertex set of a regular $n$-gon shows that
$$f_{\rm gen}(n) \leq f_{\rm conv}(n) \leq \left\lfloor \frac{n}{2} \right\rfloor.$$

Erd\H os conjectured that $f_{\rm conv}(n)= \lfloor \frac{n}{2} \rfloor$ for all $n\ge 2$. It is perfectly possible that the same equality holds for $f_{\rm gen}(n)$. The weaker statement that every set of $n$ points in convex position determines $\lfloor \frac{n}{2} \rfloor$ distinct distances was proved by Altman \cite{Al63}, \cite{Al72}. Leo Moser \cite{Mo52} proved that $f_{\rm conv}(n)\ge \frac{n}{3}$, while Szemer\'edi established essentially the same lower bound for point sets in general position: By a very simple double-counting argument, he established the inequality $f_{\rm conv}(n)\ge f_{\rm gen}(n)\ge \frac{n-1}{3}$; see \cite{E75}, \cite{PaA95}.

By combining and improving the arguments of Moser and Szemer\'edi, Dumitrescu \cite{Du06} established the bound
$$f_{\rm conv}(n) \geq \left\lceil \frac{13n-6}{36} \right\rceil.$$

In the present note, we show that Dumitrescu's bound can be further improved.

\begin{theorem}\label{thm:main}
The maximum number $f_{\rm conv}(n)$
such that any set of $n$ points in convex position in the plane
contains a point that determines at least this number of distinct distances
to the other points of the set satisfies:
$$f_{\rm conv}(n) \geq \left(\frac{13}{36}+\varepsilon\right)n - O(1),$$
for a suitable positive constant $\varepsilon$.
\end{theorem}

Our argument as presented here yields a little over $\varepsilon > 1/23000$. It is quite possible that this bound can be slightly improved through tweaks in different places, though we have abstained from doing so in the interest of simplicity.

As we shall see later, the crucial point in the argument of Szemer\'edi is to
estimate in two different ways the number of isosceles triangles
determined by the point set $P$.

Given a finite point set $P$ in the plane, we denote by $Z(P)$ the number of unordered pairs $\{(p,a),(p,b)\}$, such that $p,a,b \in P$ are three distinct points with $|pa|=|pb|$. (Here $|pa|$ stands for the length of the segment $pa$.) In other words, $Z(P)$ is the number of isosceles triangles determined by $P$, except that each equilateral triangle $abc$ is counted as {\em three} isosceles triangles (by letting each of $a$, $b$, and $c$ play the role of $p$).

If $P$ is a set of $n$ points in general
position, then it follows easily that $Z(P)\le 2\binom{n}{2}$.
Dumitrescu \cite{Du06} showed that, if $P$ is in
\emph{convex position}, then $Z(P) \le n^2(1-\frac{1}{12})$; this led him to his improved lower bound for $f_{\rm conv}(n)$.

In this paper we further improve Dumitrescu's upper bound for $Z(P)$ for $P$ in convex position. Since this is an independently interesting problem, we state it
explicitly:

\begin{problem}\label{problem:isosceles triangles}
What is the largest possible value of $Z(P)$, the number of isosceles triangles determined by $P$ (as defined above), for a planar $n$-point set $P$ in convex (or in general) position?
\end{problem}

To make our paper self-contained, in the next section we briefly sketch and later use the arguments of Moser, Szemer\'edi, and Dumitrescu. In Section~\ref{sec:new}, we prove three auxiliary results, which are then used in Section~\ref{sec:goodedges} to bound the number of isosceles triangles and finish the proof of Theorem~\ref{thm:main}.
Section~\ref{sec:last} contains some concluding remarks.

\section{The arguments of Szemer\'edi, Moser, and Dumitrescu}

First, we sketch Szemer\'edi's argument to prove the inequality $f_{\rm gen}(n) \geq \frac{n-1}{3}$. Let $P$ be a set of $n$ points in general position in the plane, and assume that for every element of $p\in P$, the number of distinct distances to the other $n-1$ points is at most $k$. Let $Z(P)$ denote the number of isosceles triangles determined by $P$, as defined above.

Clearly, we have $Z(P) \leq 2\binom{n}{2}$, because for each pair $a,b$ there exist at most two points $p \in P$ with $|pa|=|pb|$. This follows from the fact that all such points $p$ must lie on the perpendicular bisector of $ab$, and $P$ has no three collinear points. On the other hand, using the convexity of the function $\binom{x}{2}$ and Jensen's inequality, for every point $p \in P$, the number of pairs $\{a,b\}$ with $|pa|=|pb|$ is minimized when the $n-1$ points of $P \setminus \{p\}$ are distributed among the at most $k$ concentric circles around $p$ as equally as possible. That is, the number of such pairs $\{a,b\}$ is at least $k \binom{\frac{n-1}{k}}{2}$. Comparing the two bounds, we obtain
        $$nk\binom{\frac{n-1}{k}}{2}\le Z(P)\le 2\binom{n}{2},$$
which yields that $k \geq \frac{n-1}{3}$. Hence, we have
        $$f_{\rm gen}(n) \geq \frac{n-1}{3},$$
as stated.

It is obvious from the above argument that if we manage to improve the upper bound on $Z(P)$, then we obtain a better lower bound on the largest number of distinct distances measured from a point of $P$. The following lemma can be found also in \cite{Du06} and gives the relation between upper bounds for $Z(P)$ and
lower bounds for $f_{\rm conv}(n)$.

\begin{lemma}\label{lemma:improvement}
Suppose that the number $Z(P)$ of isosceles triangles determined by an $n$-point set $P$ (in general position in the plane) satisfies $Z(P)\le\alpha n^2+O(n)$ for some $\alpha\le 1$. Then $P$ contains a point from which there are at least $\frac{2-\alpha}{3}n-O(1)$ distinct distances.
\end{lemma}

\begin{proof}
Assume, as above, that for every point $p\in P$, the remaining $n-1$ points lie on at most $k$ circles centered at $p$. By Szemer\'edi's proof, we also know that $2\le \frac{n-1}k\le 3$. Otherwise, we have $k\ge n/2$, and we are done.

This means that for each point $p \in P$, the number of pairs $\{a,b\}$ with $|pa|=|pb|$ is minimized when there are precisely $k$ circles around $p$ that pass through at least one element of $P$, and each of these circles contains either 2 or 3 points. Since $n-1=(3k-n+1)2+(n-1-2k)3$, we can assume that in the worst case $3k-n+1$ circles contain 2 points and $n-1-2k$ circles contain 3 points. Thus, the number of pairs $\{a,b\}$ with $|pa|=|pb|$ is at least
$(3k-n+1)+(n-1-2k)3=2(n-1)-3k$. Therefore, $Z(P) \geq n(2(n-1)-3k)$. Combining this with the upper bound on $Z(P)\le \alpha n^2+O(n)$, the lemma follows.
\end{proof}

Dumitrescu \cite{Du06} showed that if $P$ is a set of $n$ points in {\em convex} position, then $Z(P) \leq \frac{11}{12}n^2$. Plugging this bound into Lemma \ref{lemma:improvement}, he obtained $f_{\rm conv}(n)\ge \frac{13}{36}n-O(1)$.

In the present note, we slightly improve on Dumitrescu's upper bound on $Z(P)$ for point sets in convex position, and hence on his lower bound for $f_{\rm conv}(n)$. For this, we first recall some terminology of Moser \cite{Mo52} and Dumitrescu \cite{Du06}.

\begin{definition}
A set of points $Q$ in convex position is called a \emph{cap with endpoints $a$ and $b$} if the elements of $Q$ can be enumerated in cyclic order, as $x_{1}, x_2, \ldots, x_{t}$, such that $x_1=a, x_t=b$ and there is a circle $C$ passing through $a$ and $b$
such that all $x_i$ lie in the closed region bounded by $ab$ and the shorter arc of $C$ delimited by $a$ and $b$. (If the two arcs of $C$ are of the same length, either of them will do. See Figure~\ref{fig_cap}.)
\end{definition}

\begin{figure}
\centerline{\includegraphics{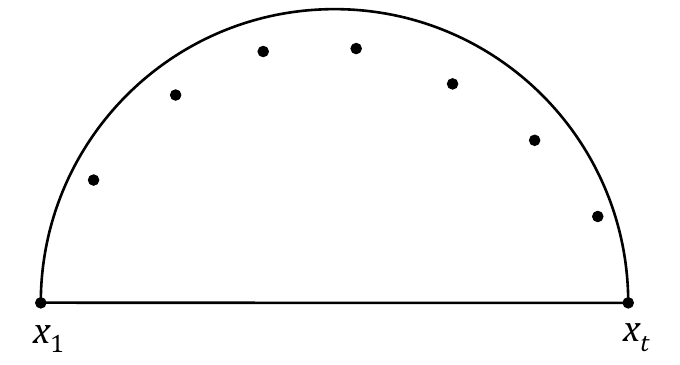}}
\caption{\label{fig_cap}A cap.}
\end{figure}

It is not hard to see (using Thales' theorem) that $x_{1}, x_2, \ldots, x_{t}$ form a cap if and only if $\mangle x_{1}x_{i}x_{t} \geq \frac{\pi}{2}$ for all $1<i<t$. Using the convexity of the point set, this is further equivalent to the condition that $\mangle x_{i}x_{j}x_{k} \geq \frac{\pi}{2}$ for every $1 \leq i < j < k \leq t$. This implies:
\begin{enumerate}
\item For every cap $x_{1}, x_2, \ldots, x_{t}$, we have
$$|x_{1}x_{2}|<|x_{1}x_{3}|< \ldots <|x_{1}x_{t}|.$$
(Indeed, since $\mangle x_{1}x_{i}x_{i+1}$ is the largest angle in the triangle $x_{1}x_{i}x_{i+1}$, we have
$|x_{1}x_{i+1}|>|x_{1}x_{i}|$.)
\item Every subset of a cap is also a cap.
\end{enumerate}

Moser \cite{Mo52} noticed that the smallest circumscribing circle around a set $P$ in convex position divides it into at most three caps that meet only at their endpoints. At least one of them has length $t\ge \lceil\frac{n}3\rceil+1$. Therefore, using property 1 above he obtained $f_{\rm conv}(n)\ge \lceil\frac{n}3\rceil$.

\begin{definition}
Let $P$ be a set of points in convex position. An unordered pair (edge) $\{a,b\}\subset P$ is called {\em good} if the perpendicular bisector of the segment $ab$ passes through at most one point of $P$. Otherwise, it is called {\em bad}. The pair (edge) $\{a,b\}$ will be often identified with the segment $ab=ba$.
\end{definition}

\begin{definition}
Let $Q \subset P$ be a cap with endpoints $a$ and $b$, and let $c,d\in Q$. A point $x \in P$ is called a {\em witness for the edge $cd$} if $x$ lies on the perpendicular bisector of the segment $cd$, and the line $ab$ does not separate $x$ from the points of $Q$.
\end{definition}

\begin{figure}
\centerline{\includegraphics{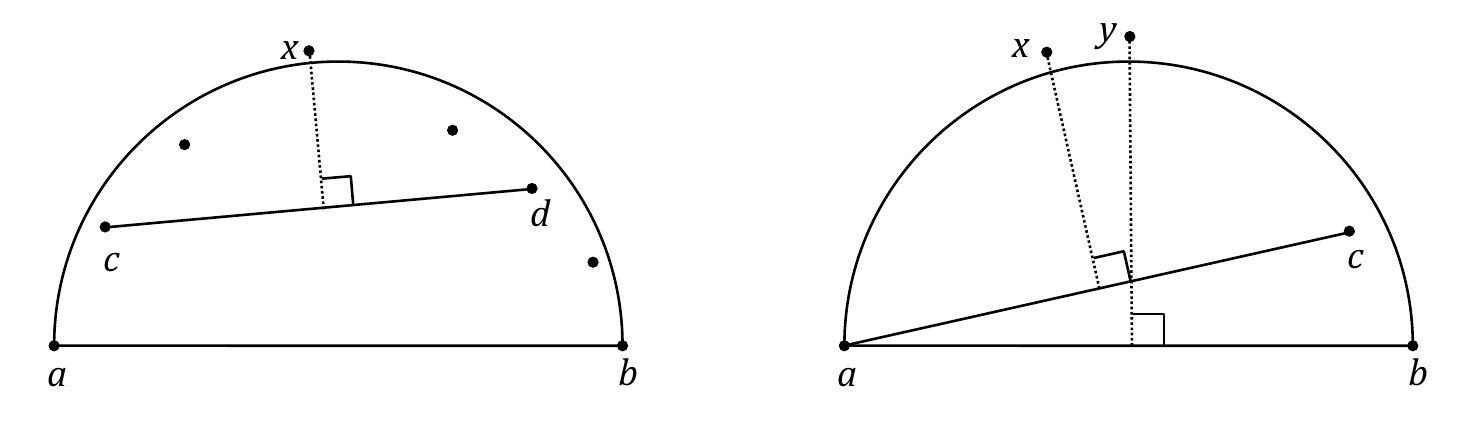}}
\caption{\label{fig_witness}Left: A witness for an edge in a cap (note that the witness does not necessarily belong to the cap). Right: Witnesses for two edges in a cap sharing a common vertex.}
\end{figure}

Since $P$ is in convex position, the witness $x$ for an edge $cd$, if exists, is uniquely determined. Furthermore, the witness $x$ for $cd$ must lie between the two points $c$ and $d$ in the circular order around $P$ from $a$ to $b$ (see Figure~\ref{fig_witness}, left).

The following lemma is a stronger version of a statement from
\cite{Du06}.

\begin{lemma}\label{lemma:monotone}
Let $Q\subset P$ be a cap with endpoints $a$ and $b$. Let $c$ be a point of $Q$ and assume that $x$ and $y$ from $P$ are witnesses for $ac$ and $ab$, respectively. Then $x$ lies between $a$ and $y$ in the circular order around $P$ from $a$ to $b$. In particular, we have $x \neq y$. See Figure~\ref{fig_witness}, right.
\end{lemma}

\begin{proof}
Assume without loss of generality that $ab$ is horizontal, $a$ is to the left of $b$, and $Q$ lies above the line $ab$. Assume to the contrary that $y$ lies between $a$ and $x$, or that $y=x$. We know already
that $x$ lies between $a$ and $c$. We have $|yc|\geq |ya|=|yb|$. Therefore, we have $\mangle ycb < \pi/2$. However, we know that $\mangle acb < \mangle ycb$, contradicting the fact that $\mangle acb \geq \pi/2$, as $Q$ is a cap.
\end{proof}

\begin{corollary}[\cite{Du06}]\label{corollary:half}
Let $Q$ be a cap consisting of $t$ points. Then there are at most $\frac{1}{4}{t^2}$ edges in $Q$ that have a witness in $Q$.
\end{corollary}

\begin{proof}
Denote the points of $Q$ in cyclic order by $x_{1}, x_2, \ldots, x_{t}$.
By Lemma \ref{lemma:monotone}, no two edges of $Q$ that share a common vertex can have the same witness. Therefore, $x_{i}$ can witness at most $\min (i-1, n-i)$ edges in $Q$. Hence, the number of edges in $Q$ with a witness in $Q$ is at most $2(1+2+ \ldots + \lceil \frac{t}{2}-1 \rceil) \leq \frac{1}{4}t^2$.
\end{proof}

\begin{corollary}\label{corollary:dumitrescu}{\rm \cite{Du06}}
Let $P$ be a set of $n$ points in convex position in the plane.
Then $P$ has at least $\frac{n^2}{12}$ good edges.
\end{corollary}

\begin{proof}
The smallest enclosing circle $C$ of $P$ passes through (at most) 3 points $a,b,c\in P$ (possibly not all distinct) such that each of the  arcs delimited by them is at most a semi-circle. Thus, $a,b,$ and $c$ divide $P$ into at most 3 caps.

If $a,b,c$ are distinct, let $r,s$, and $t$ denote the number of points in these caps, where $r+s+t=n+3$. By Corollary \ref{corollary:half} and the Cauchy--Schwarz inequality, the total number of good edges completely contained in one of the caps is at least
$$\frac{1}{4}(r^2+s^2+t^2)\ge \frac{1}{4} 3\frac{n^2}{9}=\frac{n^2}{12}.$$
If $b=c$, say, we obtain an even better lower bound.
\end{proof}

In order to complete Dumitrescu's argument, notice that if $xy$ is a good edge in $P$, then there is at most one isosceles triangle with base $xy$. Thus, we have
$$Z(P) \leq 2\binom{n}{2}-\#\{\mbox{good edges}\}.$$
According to Corollary~\ref{corollary:dumitrescu}, this implies $Z(P) < (11/12)n^2$. Plugging this bound into Lemma~\ref{lemma:improvement}, Dumitrescu obtained (\cite{Du06}) that $P$ determines at least
$(13/36)n-O(1)$ distinct distances.

Our improvement on Dumitrescu's argument is as follows:

\begin{theorem}\label{theorem:ma}
Let $P$ be a set of $n$ points in convex position. Then $P$ has at least $\alpha n^2$ good edges, where $\alpha=1/11.981$, and therefore, $Z(P) \le (10.981/11.981) n^2$.
\end{theorem}

Theorem~\ref{theorem:ma} and Lemma~\ref{lemma:improvement} yield
\begin{equation*}
f_{\rm conv} (n) \ge \left(\frac{13}{36} + \frac{1}{22701}\right) n - O(1),
\end{equation*}
proving Theorem~\ref{thm:main}.

\section{Three lemmas on witnesses}\label{sec:new}

To improve the lower bound $f_{\rm conv}(n) \geq \frac{13}{36}n$, we need a couple of auxiliary results. The first such statement is a simple consequence of Lemma \ref{lemma:monotone}.

\begin{lemma}\label{cor:half_easy}
Let $Q$ be a cap of size $t$ with endpoints $a$ and $b$. Then the total number of edges adjacent to $a$ and $b$ with no witness in $Q$ is at least $t-1$.
\end{lemma}

\begin{proof}
Let $x$ be the witness in $Q$ for the edge $ab$, if it exists; otherwise add such a point
$x$ keeping $Q \cup\{x\}$ in convex position.

By Lemma \ref{lemma:monotone}, all witnesses to edges adjacent to $b$
are between $b$ and $x$, while all witnesses for edges adjacent to $a$
are between $a$ and $x$. We know already that a point in $Q$ can be a witness
for at most one edge adjacent to $a$ and at most one edge adjacent to $b$.
We conclude that every point in $Q \setminus \{a,b\}$ may be a witness for
at most one edge adjacent to $a$ or to $b$.
As there are $2t-3$ edges in total that are adjacent to $a$ or to $b$,
there must be at least $2t-3-(t-2)=t-1$ of them with no witness in $Q$.
\end{proof}

The following geometric lemma, which is of independent interest, will be a crucial element of our proof.

\begin{lemma}\label{lemma:tech}
Let $Q = \{a,b,c,d,e\}$ be a cap, with the points appearing in clockwise order, such that $c$ is a witness for $ae$ and $d$ is a witness for $be$. Then $|ab|>|cd|$.
\end{lemma}

\begin{proof}
First, we can assume without loss of generality that $b$ lies on the segment $ac$;\footnote{Then $Q$ is not in convex position anymore, but only in weakly convex position. But, as we show, even with this relaxation the claim holds.} for otherwise, we can slide $b$ counterclockwise along the circle centered at $d$ passing through $b$, until $b$ reaches $ac$, and this only decreases $|ab|$ (in fact, $|ab|$ keeps decreasing until $b$ reaches the segment $ad$).

Next, let $o$ be the midpoint of $ae$, let $C$ be the circle centered at $o$ passing through $a$ and $e$, and let $\ell$ be the line passing through $d$ perpendicular to $be$. Without loss of generality we can slide $d$ along $\ell$ either inwards or outwards, making sure $Q$ is still a cap, so as to maximize $|cd|$. Then, $d$ falls in one of these three cases (see Figure~\ref{fig_tech_3cases}):

\begin{figure}
\centerline{\includegraphics{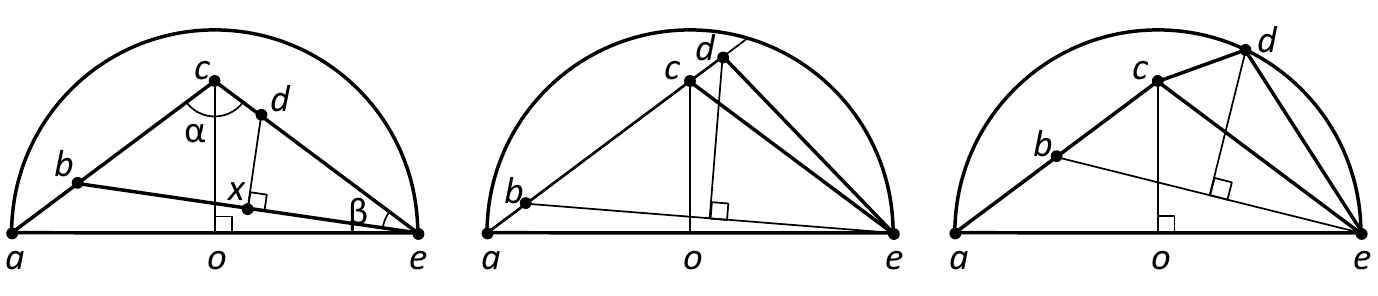}}
\caption{\label{fig_tech_3cases}The three cases for Lemma~\ref{lemma:tech}.}
\end{figure}

\begin{enumerate}
\item $d$ lies on $ce$.
\item $d$ lies on the line through $a$ and $c$.
\item $d$ lies on $C$.
\end{enumerate}

Suppose the first case. Let $x = \ell \cap be$, and let $\alpha = \mangle ace$ and $\beta = \mangle bec$. Then, by applying the law of sines on the triangle $bce$ and considering the right-angled triangle $xde$, we get $|bc| / \sin \beta = |be| / \sin \alpha = 2|de|\cos\beta/\sin\alpha$, so $|bc| = |de|\sin(2\beta)/\sin\alpha = |de|\sin(2\beta)/\sin(\pi-\alpha)$. But $\pi-\alpha = 2\mangle cea > 2\beta$, which implies that $|bc| < |de|$, or equivalently $|ab| > |cd|$.

Now suppose the second case. Then $|ab|>|cd|$ is equivalent to $|ac|>|bd|$. But $|ac|=|ce|$ and $|bd|=|de|$. Furthermore, $|ce|>|de|$ since $\mangle cde \ge \pi/2$ in the triangle $cde$, so we are done.

The third case is divided into two subcases, according to whether $d$ lies higher or lower than $c$ (meaning, whether $\mangle dco$ is obtuse or acute).

If $d$ lies higher than $c$, then without loss of generality we can move $c$ down towards $o$, and move $b$ counterclockwise along the circle centered at $d$, until both $c$ and $b$ reach the segment $ad$. This only decreases $|ab|$ and increases $|cd|$, and we fall back into case 2.

Finally, suppose that $d$ lies on $C$ but lower than $c$. Let $b'$ be the point along $ad$ satisfying $|b'd| = |bd|$; and let $c'$ be the intersection point of $C$ and the ray $\overrightarrow{oc}$. Note that $|ab| \ge |ab'|$ and $|cd| \le |c'd|$. We show algebraically that $|ab'| > |c'd|$, which proves our claim:

Suppose without loss of generality that $o$ is the origin and $C$ has radius $1$. Let $d = (x, \sqrt{1-x^2})$ for some real number $0<x<1$. Then $|ab'| = |ad| - |b'd| = \sqrt{2+2x} - \sqrt{2-2x}$, while $|c'd| = \sqrt{2-2\sqrt{1-x^2}}$, and a routine algebraic calculation shows that $|ab'| > |c'd|$ for all $0<x<1$.
\end{proof}

\begin{lemma}\label{lemma:sequence}
Let $Q = \{x_1, x_2, \ldots, x_{2t}\}$ be a cap with the points appearing in that order.
Then the number of edges between the sets
$\{x_{1}, x_2, \ldots, x_{t}\}$ and $\{x_{t+1},x_{t+2}, \ldots, x_{2t}\}$ that have a witness in $Q$ is at most $\frac{7}{8}t^2 + O(t)$.
\end{lemma}

\begin{proof}
Let $G$ denote the geometric graph whose vertices are the points of $Q$ and whose edges are those edges whose number we wish to bound in the lemma. Consider the set of segments $\{x_ix_{i+1} \mid i\neq t , 1 \leq i \leq 2t-1\}$, and let $s_1, s_2, \ldots, s_{2t-2}$ denote these segments enumerated in increasing order of their
lengths (i.e., we have $|s_1| \le |s_2| \le \cdots \le |s_{2t-2}|$).
For every $1 \leq i \leq 2t-2$, denote by $u_i$ and $v_i$ the endpoints of $s_i$ so that $u_i = x_j$ and $v_i = x_{j+1}$ for some $j$.

We claim that $d(u_i) + d(v_i) \leq t+i$ for $i=1,2, \ldots, t$, where $d(v)$ is the degree of vertex $v$ in $G$.

Indeed, fix $1 \leq i \leq t$. Suppose without loss of generality that
$v_i = x_j$ for some $j \leq t$. Let $x_k$ be a vertex with $k>t$ such that
both $u_ix_k$ and $v_ix_k$ are in $G$. Let their witnesses be
$x_\ell$ and $x_{\ell'}$, respectively, with $\ell'>\ell$. Then, by
Lemma~\ref{lemma:tech} (taking $a=u_i$, $b=v_i$, $c=x_\ell$, $d=x_{\ell'}$, $e=x_k$), we have
$|x_{\ell}x_{\ell+1}| \le |x_{\ell} x_{\ell'}| < |s_i|$.
Therefore, either $\ell = t$ or $x_{\ell}x_{\ell+1} = s_{i'}$ for some $i'<i$. Obviously, there are only $i-1$ such segments $s_{i'}$. Furthermore, by Lemma~\ref{lemma:monotone}, different edges $u_ix_k$, $u_ix_{k'}$ must have different witnesses.
It follows that there can be at most $i$ vertices among $x_{t+1}, \ldots, x_{2t}$ that are connected to \emph{both} $u_i$ and $v_i$ in $G$, and therefore $d(u_i)+d(v_i) \le t+i$, as claimed.

Adding up over all segments, we obtain
\begin{align*}
4|E(G)| - 4t = 2\sum_{i=1}^{2t} d(x_i) - 4t &\le \sum_{i=1}^{2t} d(x_i) - \bigl( d(x_1)+d(x_t)+d(x_{t+1})+d(x_{2t}) \bigr)\\
&= \sum_{i=1}^{2t-2} (d(u_i)+d(v_i))\\
&\le \bigl((t+1) + \cdots + 2t \bigr) + 2t + \cdots + 2t = \frac{7}{2}t^2 - \frac{7}{2}t,
\end{align*}
and the lemma follows.
\end{proof}

The bound in Lemma~\ref{lemma:sequence} can be slightly improved, though we have abstained from doing so in the interest of simplicity, and because even a tight bound for the lemma would only yield a small improvement for $\varepsilon$ in Theorem~\ref{thm:main}. Lev and Pinchasi~\cite{LeP12} have shown that Lemma~\ref{lemma:sequence} cannot be improved beyond $\frac{3}{5}t^2 + O(t)$; see further discussion in Section~\ref{sec:last}.

\section{Proof of the main result}\label{sec:goodedges}

\begin{proof}[Proof of Theorem~\ref{theorem:ma}]
Let $p_1, \ldots, p_n$ be the points of $P$ in circular order. In this proof, by the \emph{circular distance} between two points $p_i,p_j \in P$ we mean the minimum of $(j-i)\bmod n$ and $(i-j)\bmod n$.

Fix two constants $0<a,d<1$ to be determined later. We will think of $d$ as much
smaller than $a$.
Perform the following $dn$ steps: Let $P_{1}=P$. At step $i$ choose the smallest
enclosing circle of $P_{i}$ and let $x_{i},y_{i},z_{i}$ be the three points of
$P_{i}$ that lie on this circle. (If the circle passes through only two points, let $x_i=y_i$.) Then let $P_{i+1}=P_{i} \setminus \{x_{i},y_{i},z_{i}\}$.

\begin{figure}
\centerline{\includegraphics{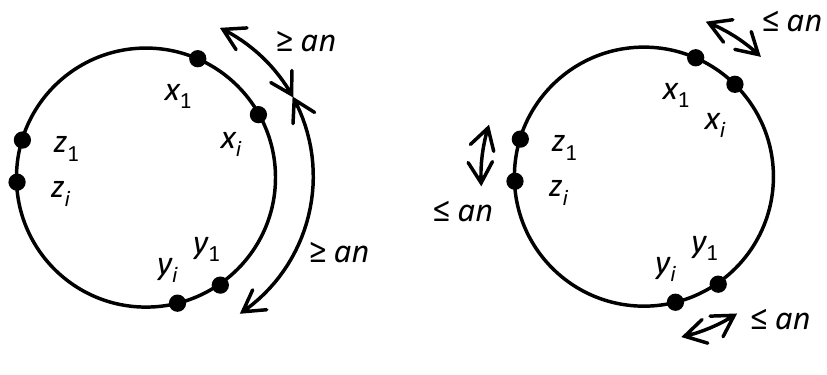}}
\caption{\label{fig_twocases}The two cases in the proof of Theorem~\ref{theorem:ma}.}
\end{figure}

We now consider two cases (see Figure~\ref{fig_twocases}):

\noindent {\bf Case 1.}
For some index $i$, $1<i\le dn$, some point among $x_i$, $y_i$, and $z_i$ is at circular distance at least $an$ from each of the three points $x_1$, $y_1$, and $z_1$. Without loss of generality let $x_i$ be that point.

Note that $P_i$ is partitioned into three caps in two different ways: The points $x_1$, $y_1$, $z_1$ define caps $Q_1$, $Q_2$, $Q_3$, while the points $x_i$, $y_i$, $z_i$ define caps $Q'_1$, $Q'_2$, $Q'_3$.
The intuition for this case is that, since $x_i$ is far from $x_1$, $y_1$, $z_1$, these two partitions are, in a sense, significantly different.

Applying the argument of Corollary~\ref{corollary:dumitrescu} to the caps $Q'_1$, $Q'_2$, $Q'_3$, we find at least $(n-3i)^2/12 \ge (n-3dn)^2/12$ edges that are good in $P_i$ and connect points within the same cap. However, not all of these edges are necessarily good in $P$, since the points in $P\setminus P_i$ might invalidate some of these edges.

However, by Lemma~\ref{lemma:monotone}, no point of $P\setminus P_i$ can invalidate two adjacent edges in the same cap, so each point of $P\setminus P_i$ invalidates at most $n/2$ of these edges. Thus, we are left with at least
\begin{equation*}
\frac{1}{12}(n-3dn)^2 - \frac{3}{2}dn^2
\end{equation*}
edges that are good in $P$ and are \emph{internal} to $Q'_1$, $Q'_2$, or $Q'_3$.

In addition, the $2an$ points of $P$ at circular distance at most $an$ from $x_i$ are all contained in the same cap $Q_1$, $Q_2$, or $Q_3$. Applying Lemma~\ref{lemma:sequence} to them, we find another $a^2n^2/8$ good edges in $P$ which were not counted previously, since they straddle two different caps among $Q'_1$, $Q'_2$, $Q'_3$.

Hence, in Case 1, we find at least
\begin{equation*}
\frac{1}{12}(n-dn)^2+\frac{1}{8}a^2n^2-\frac{3}{2}dn^2
\end{equation*}
good edges in $P$.

\noindent {\bf Case 2.}
For every index $i$ between $1$ and $dn$, each of $x_{i}$, $y_{i}$, $z_{i}$ is at circular distance
at most $an$ from one of $x_{1}$, $y_{1}$, $z_{1}$ in $P$.

In this case, the analysis is somewhat different.
For each $i$, we apply Lemma~\ref{cor:half_easy} three times, on the caps delimited by $(x_i,y_i)$, $(y_i,z_i)$, and $(z_i,x_i)$. It follows that the points
$x_{i},y_{i}$, and $z_{i}$ together are adjacent to
at least $n-3i$ good edges in $P_{i}$. As before, not all of these edges are necessarily good in $P$. However, by applying Lemma~\ref{lemma:monotone} on each of these three caps, we conclude that each point in $P\setminus P_i$ can invalidate at most one such edge. Hence, we are left with at least $n-6i$ good edges in $P$.

Now consider $P_{dn}$, and consider its partition into three caps $Q_1$, $Q_2$, $Q_3$ by the original three points $x_1$, $y_1$, $z_1$. By Corollary~\ref{corollary:dumitrescu}, there are at least $(n-3dn)^2/12$ edges that are good in $P_{dn}$ and connect two points within the same cap $Q_1$, $Q_2$, or $Q_3$.

As before, not all of these edges are necessarily good in $P$, but we can bound the number of edges invalidated by the points $x_i$, $y_i$, $z_i$, $i<dn$ of $P\setminus P_i$: Each such point is within circular distance $an$ of $x_1$, $y_1$, or $z_1$, so by Lemma~\ref{lemma:monotone}, it can invalidate at most $an$ of these edges.

Therefore, in Case 2, we find at least
\begin{equation*}
\sum_{i=1}^{dn}(n-6i)+\frac{1}{12}(n-3dn)^2-3dan^2=
\frac{1}{12}n^2+\frac{1}{2}dn^2-
\frac{9}{4}d^2n^2-3dan^2 - O(n)
\end{equation*}
good edges in $P$.

If we choose $a$ and $d$ properly, we can guarantee that in all cases we have
strictly more than $n^2/12$ good edges.
The values $a = 1/8.8$ and $d = 1/1132$ are close to the optimal ones, and they yield at least $n^2/11.981$ good edges.\end{proof}

\section{Concluding remarks}\label{sec:last}

\paragraph{A point set with many isosceles triangles.} Take $n-1$ points $x_{1}, \ldots, x_{n-1}$ evenly distributed on say, a quarter of a circle, together with the center of the circle $x_n$. The resulting $n$-point set $P$ is in convex position, and $Z(P)\ge 3n^2/4-O(n)$. Hence the method described in Lemma \ref{lemma:improvement} cannot yield a lower bound better than  $5n/12-O(1)$ for $f_{\rm conv}(n)$.

\paragraph{Bichromatic arithmetic triples.} The following combinatorial question was motivated by our study of Lemma~\ref{lemma:sequence}:

\begin{problem}\label{problem:arith3}
Let $R$ be a set of $t$ red negative numbers and $B$ be a set of $t$ blue positive numbers. What is the maximum number of triples in arithmetic progression in $R\cup B$ that are bichromatic?
\end{problem}

The argument in the proof of Lemma~\ref{lemma:sequence} yields an upper bound of $\frac{7}{8}t^2 + O(t)$ for Problem~\ref{problem:arith3}. Lev and Pinchasi~\cite{LeP12} have recently solved Problem~\ref{problem:arith3}, showing that the answer is $\frac{3}{5}t^2 \pm O(t)$. Their upper bound for the problem does not translate into an improved upper bound for the lemma. However, their lower bound can be transformed into a point set on a circular arc, which shows that Lemma~\ref{lemma:sequence} cannot be improved beyond $\frac{3}{5}t^2 - O(t)$.

\paragraph{Acknowledgements} We would like to thank the referees for their helpful comments.

\end{document}